\documentclass{article}
\usepackage{spconf,amsmath,graphicx}
\usepackage{amsthm}
\usepackage{amssymb}
\usepackage{amsfonts}
\usepackage{bbding}
\usepackage{epsfig}
\usepackage{cite}
\usepackage{color, soul}
\usepackage{array}
\usepackage{multirow}
\usepackage{booktabs}

\newtheorem{Proposition}{Proposition}

\title{Quantization Reference Voltage of the Modulated Wideband Converter}

\name{Yaming Wang, Laming Chen and Yuantao Gu$^*$\thanks{This work
was supported in part by the National Natural Science Foundation of
China under Grants NSFC 60872087 and NSFC U0835003. The
corresponding author of this paper is Yuantao Gu
(gyt@tsinghua.edu.cn). }}
\address{State Key Laboratory on Microwave and Digital Communications\\
 Tsinghua National Laboratory for Information Science and Technology\\
 Department of Electronic Engineering, Tsinghua University, Beijing 100084, CHINA}

\begin{document}
\ninept
\maketitle

\begin{abstract}
The Modulated Wideband Converter (MWC) is a recently proposed analog-to-digital converter (ADC) based on
Compressive Sensing (CS) theory. Unlike conventional ADCs, its quantization reference voltage, which
is important to the system performance, does not equal the maximum amplitude of original analog signal. In
this paper, the quantization reference voltage of the MWC is theoretically analyzed and the conclusion demonstrates that the reference voltage is proportional to the square root of $q$, which is a trade-off parameter between sampling rate and number of channels. Further discussions and simulation results show that the reference voltage is proportional to the square root of $Nq$ when the signal consists of $N$ narrowband signals.
\end{abstract}


\section{Introduction}
\label{sec:intro}

In telecommunication, radar and other fields of electronic engineering, it is usually of great
concern about multiband signals, which consist of several narrowband signals modulated by different
carrier frequencies. To acquire such a signal, a common practical method \cite{ref1} is to
demodulate the narrowband signals by their carrier frequencies to baseband, respectively. After
low-pass filtered to reject frequencies from other bands, the narrowband signals are converted to
digital samples at rate corresponding to their actual bandwidths. Thus each band is acquired
separately and the total sampling rate is the sum of the bandwidths. This method can reach the
minimal sampling rate while the carrier frequencies must be known a priori.

When the carrier frequencies are unknown, efficient sampling of the multiband signal is a
challenging task because the Nyquist frequency of the signal may exceed the capability of the
up-to-date analog-to-digital converters (ADCs). In this scenario, several sub-Nyquist sampling
strategies have been proposed to handle this challenge. In \cite{ref2} and \cite{ref3}, a scheme
called multicoset sampling is proposed to perform analog-to-digital conversion at low rate, but
knowledge of the frequency support must be known for recovery. Periodic nonuniform sampling is
another sub-Nyquist sampling strategy where several ADCs are applied to form a high sampling rate
for the multiband signal \cite{ref5}.

In recent years, Compressive Sensing (CS) has been raised and developed to recover sparse signals
from far fewer samples \cite{ref6}. Due to the fact that multiband signal is sparse in frequency domain,
several novel analog-to-digital architectures have been designed based on CS theory. In Random
Demodulator \cite{ref7, ref8}, the original signal is multiplied by a pseudorandom sequence,
integrated and sampled at sub-Nyquist rate. In the recovery stage, CS algorithms are performed to
recover the original signal from these samples.

The Modulated Wideband Converter (MWC) \cite{ref9, ref11, ref12} is another CS-based
sub-Nyquist sampling system. It also consists of two stages: sampling, as depicted in
Fig.~\ref{fig1}, and reconstruction. In the sampling stage, modulated sampling is conducted by
mixers and lowpass filters in multiple channels. In the reconstruction stage, sparsity constraint
algorithms are applied to recover the original signal from the multiple observations. The sampling
rate $f_s$ and the frequency $f_p$ of the periodic mixing function $p_i(t)$ satisfies $f_s=qf_p$
with odd $q$, which is a strategy collapsing $q$ channels to a single one. The strategy allows
designers to trade off between sampling rate and the number of hardware devices. Recent researches
focused on this system include a calibration system with simple structure and low computational
complexity to obtain actual measurement matrix \cite{ref13}.

\begin{figure}[t]
\centering
\includegraphics[width=0.4\textwidth]{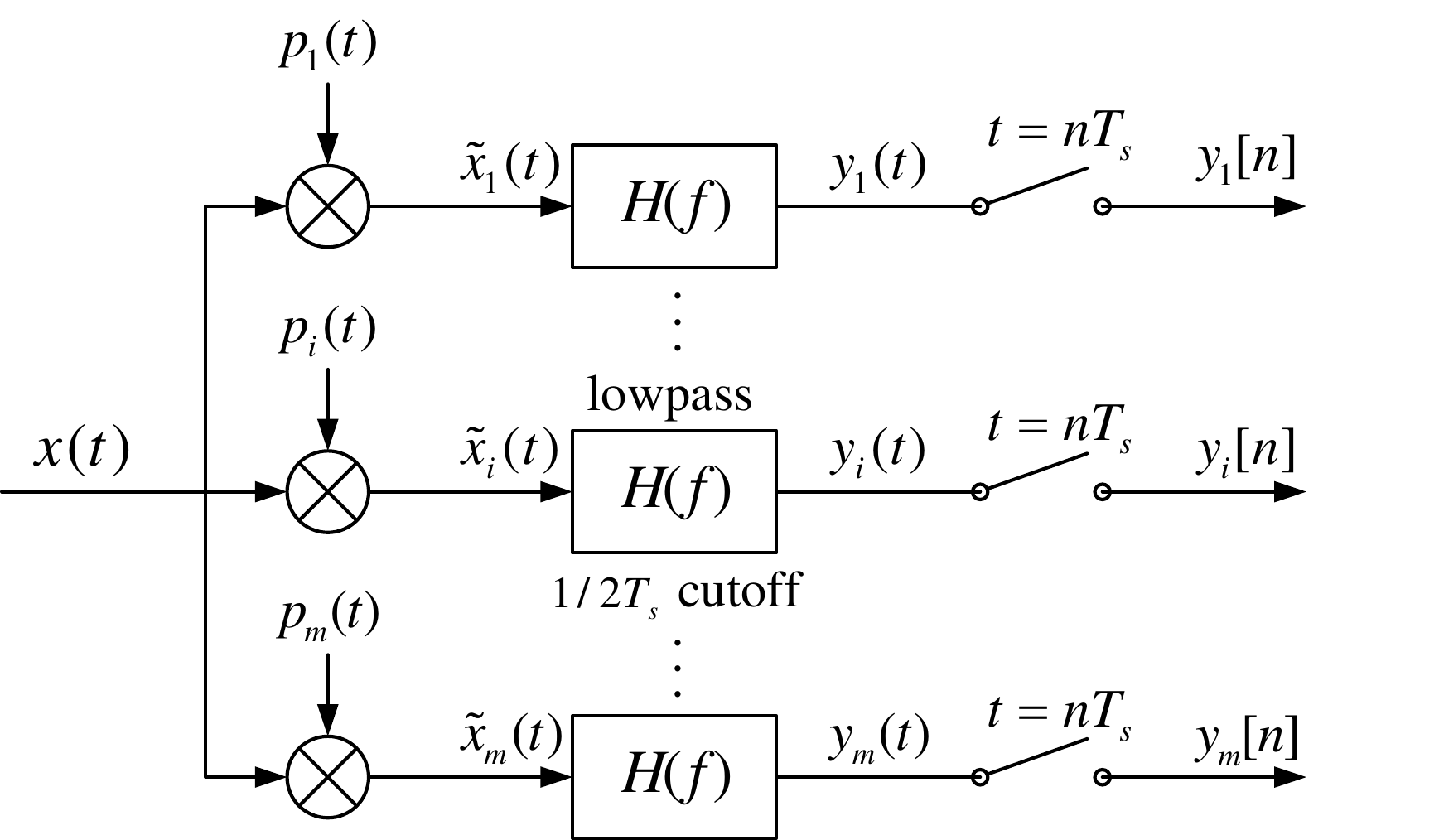}
\caption{A block diagram of the sampling stage of the Modulated Wideband Converter}\label{fig1}
\end{figure}

In order to fully use the accuracy of an ADC, the quantization reference voltage is an essential parameter to consider about. The reference voltage of an conventional ADC simply equals the maximum value of the input analog signal. In the MWC, however, the quantization takes place at the end of the sampling stage and the maximum of the samples is sensitive to several parameters, especially the randomness of the mixing functions. Thus the reference voltage can not be derived from the input signal intuitively. Moreover, simulation results exploit that the maximum of the samples varies widely even when the input signal is fixed. Therefore, an improper reference voltage may have a severe impact on the performance of the whole system. Though important in practice, few researches have been focused on this problem and the reference voltage in practical implementation is obtained by experiments.

This paper aims to theoretically study the reference voltage of quantizer in the MWC. A reasonable method to acquire reference voltage is proposed. By analyzing the input signal under several assumptions, we draw the conclusion that the reference voltage is proportional to the square root of the collapse parameter $q=f_s/f_p$. Further discussions show
that the assumptions made in the analysis are acceptable and that the conclusion is valid for
arbitrary multiband signals.

This paper is organized as follows. Section 2 briefly introduces the structure of the MWC and related mathematical expressions. Section 3 studies the reference voltage and come to the conclusion under rational assumptions. The theoretical results in Section 3 are numerically validated in Section 4.

\section{The Modulated Wideband Converter}
\label{sec:MWC}

As can be seen from Fig.~\ref{fig1}, the input signal $x(t)$ enters $m$ channels simultaneously. In
the $i$-th channel, it is multiplied by a $T_p$-periodic pseudorandom sequence $p_i(t)$, whose
frequency is $f_p=1/T_p$. Specifically, $p_i(t)$ is chosen as a piecewise constant function that
randomly alternates between the levels $\pm 1$ for each of $M$ equal time intervals. The mixed
signal is then truncated by an ideal lowpass filter with cutoff $1/(2T_s)$. Finally, the filtered
signal $y_i(t)$ is sampled at rate $f_s=1/T_s$. According to \cite{ref11}, the parameters $f_s$ and
$f_p$ satisfy $f_s=qf_p$ with odd $q$. Further assume that $x(t)$ is bandlimited to $[-f_{\rm
NYQ}/2,f_{\rm NYQ}/2]$, where $f_{\rm NYQ}$ is much larger than the total sampling rate of the
system.

Formally, the mathematical description of $p_i(t)$ is
\begin{align}
p_i(t)=\alpha_{ik},\quad &k\frac {T_p} M \le t < (k+1) \frac {T_p} M, \nonumber \\
&0\le k \le M-1, \nonumber
\end{align}
with $\alpha_{ik}\in\{+1,-1\}$ and $p_i(t+nT_p)=p_i(t)$ for every $n\in\mathbb Z$. Since $p_i(t)$
is $T_p$-periodic, it has the Fourier expansion
$$
p_i(t)=\sum_{l=-\infty}^{+\infty}c_{il}{\rm e}^{{\rm j}\frac {2\pi} {T_p} lt},
\vspace{-1em}
$$
where $\{c_{il}\}$ are Fourier coefficients. Define $\phi={\rm e}^{-{\rm j}2\pi/M}$, thus
\begin{equation}\label{eq12}
c_{il}=d_l\sum_{k=0}^{M-1}\alpha_{ik}\phi^{lk},
\vspace{-1em}
\end{equation}
where
$$
d_l=\frac 1 {T_p} \int_0^{\frac{T_p} M}{\rm e}^{-{\rm j}\frac {2\pi} {T_p} lt}\,{\rm d}t.
$$

In \cite{ref11}, it has been proved that the Fourier transform of $y_i(t)$ is expanded as
\begin{equation}\label{eq1}
Y_i(f)=\sum_{l=-L_0}^{L_0}c_{il}X(f-lf_p),\quad f\in\left[-f_s/2,f_s/2\right],
\vspace{-1em}
\end{equation}
where
$$
L_0=\left\lceil \frac {f_{\rm NYQ}+f_s} {2f_p} \right \rceil-1.
$$

\section{Reference Voltage of Quantization}
\label{REF}

\subsection{Preliminary}

We first consider about a basic scenario where the signal is only composed of one modulated
narrowband signal, and then discuss about arbitrary multiband signals. In the sequel, the
Fourier transform of $x(t)$, denoted by $X(f)$, consists of two symmetric bands centered at $\pm
l_0f_p$ respectively, where $(l_0-1/2)f_p\ge f_s/2$. The width of both bands satisfies $B\le f_p$. Formally,
\begin{equation}\label{eq2}
X(f)=X_0(f-l_0f_p)+X_0(f+l_0f_p),
\end{equation}
where
$$
X_0(f)=0,\quad f\notin\left[-f_p/2,f_p/2 \right].
$$

Notice that the structure of each channel is identical, and that the
maximum sample of $y_i[n]$ is no larger than the maximum value of
analog signal $y_i(t)$. As a result, the sampling stage can be
simplified to Fig.~\ref{fig2} when reference voltage is being
analyzed. For convenience, in the sequel, the subscript $i$ which denotes the channel number is omitted.

\begin{figure}[t]
\centering
\includegraphics[width=0.35\textwidth]{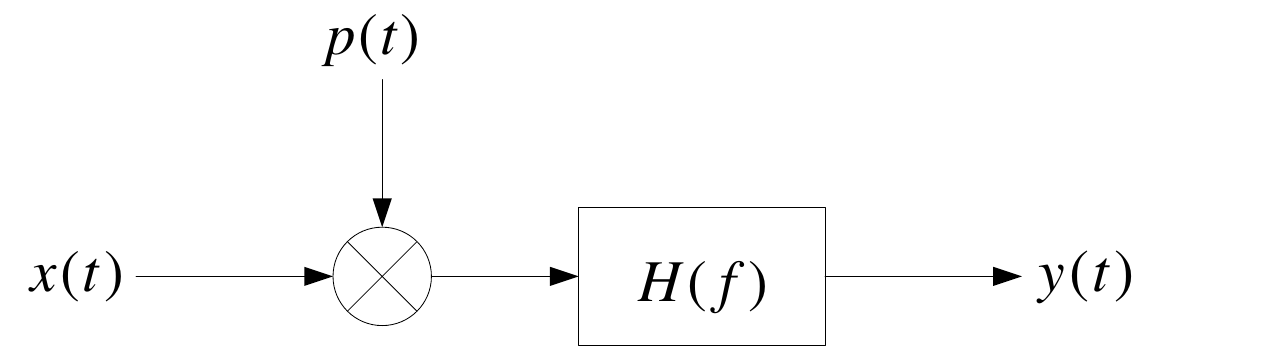}
\caption{The simplified sampling stage}\label{fig2}
\end{figure}

For the basic scenario (\ref{eq2}), the series on the right side of
(\ref{eq1}) have only $2q$ nonzero terms. Then
\begin{align}\label{eq3}
Y(f)=\sum_{l=l_1}^{l_1+q-1} \left[c_l X(f-l f_p)+c_{-l}X(f+l f_p)\right],\\
 \quad f\in\left[-f_s/2, f_s/2\right], \nonumber
\end{align}
where $l_1=l_0-(q-1)/2$. Fig.~\ref{fig3} exploits the simplification
from (\ref{eq1}) to (\ref{eq3}) intuitively.

\begin{figure}[t]
\centering
\includegraphics[width=0.4\textwidth]{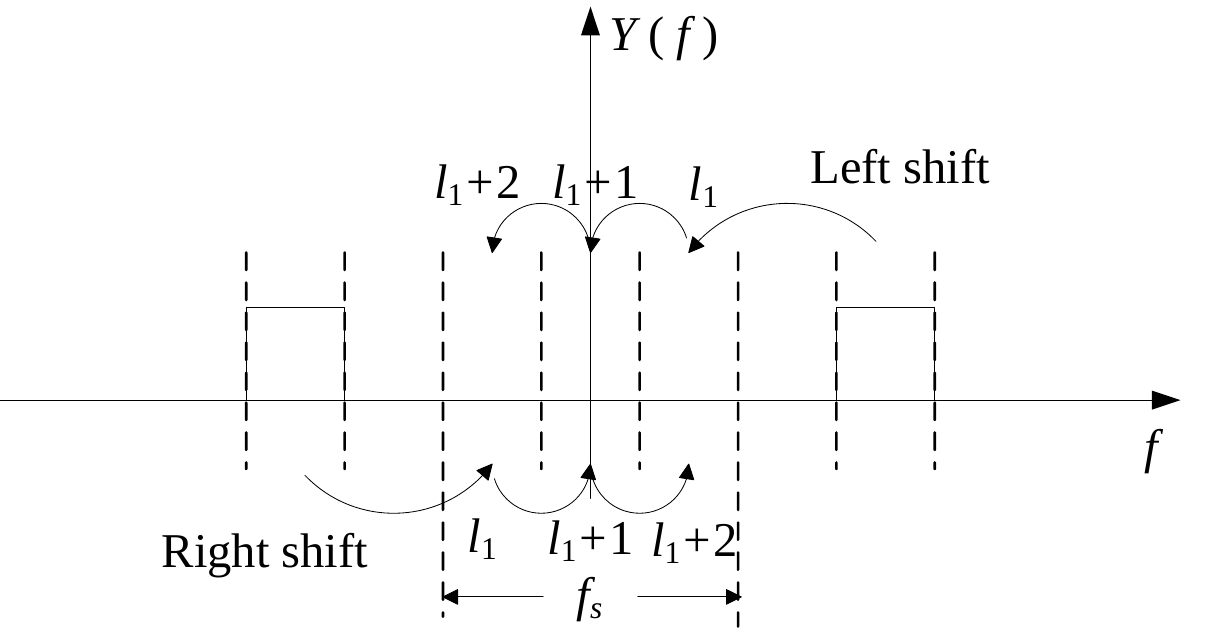}
\caption{The simplification from (\ref{eq1}) to (\ref{eq3})
(Parameters: $q=3$, $l_0=3$, $l_1=2$).}\label{fig3}
\end{figure}

Considering (\ref{eq3}) in time domain, one has
\begin{align}\label{eq4}
y(t)=&\sum_{l=l_1}^{l_1+q-1}c_l\int_{-f_s/2}^{f_s/2}X(f-lf_p){\rm e}^{{\rm j}2\pi
ft}\,{\rm d}t \nonumber\\
&+\sum_{l=l_1}^{l_1+q-1}c_{-l}\int_{-f_s/2}^{f_s/2}X(f+lf_p){\rm e}^{{\rm j}2\pi
ft}\,{\rm d}t.
\end{align}
Since $[-f_s/2,f_s/2]$ contains only one sideband of $X(f-lf_p)$,
Hilbert transform can be applied to represent such single sideband
signal. Setting ${\hat x}(t)=x(t)* \left\{ 1/(\pi t) \right\}$ as Hilbert
transform of $x(t)$, ${\hat X}(f)$ denotes the Fourier transform of
${\hat x}(t)$. According to the definition of Hilbert transform,
$[X(f)-{\rm j}{\hat X}(f)]/2$ represents the negative-frequency band of
$X(f)$. Hence
\begin{align}\label{eq5}
&\sum_{l=l_1}^{l_1+q-1}c_l \int_{-f_s/2}^{f_s/2} X(f-l f_p){\rm e}^{{\rm j}2\pi ft}\,{\rm d}t \nonumber \\
=&\sum_{l=l_1}^{l_1+q-1}c_l {\rm e}^{{\rm j}2\pi lf_pt} \int_{-\infty}^{+\infty} \frac {X(f)-{\rm j}{\hat X}(f)} 2 {\rm e}^{{\rm j}2\pi ft}\,{\rm d}t \nonumber \\
=&\sum_{l=l_1}^{l_1+q-1}c_l {\rm e}^{{\rm j}2\pi lf_pt} \frac {{\tilde x}^*(t)} 2,
\end{align}
where ${\tilde x}(t)=x(t)+{\rm j}{\hat x}(t)$ is the analytic
signal of $x(t)$.

Dealing with the second term on the right side of (\ref{eq4}) in the
same way, it can be derived that
\begin{align}\label{eq6}
\left|y(t)\right|=&\left|\sum_{l=l_1}^{l_1+q-1}\Re\left\{c_l{\rm e}^{{\rm j}2\pi lf_pt}{\tilde x}^*(t)\right\}\right| \nonumber\\
=&\left|{\tilde x}(t)\right| \left|\sum_{l=l_1}^{l_1+q-1} \left|c_l\right|\cos\left(\theta_l+2\pi lf_pt-\theta_{{\tilde x}}(t)\right)\right|,
\end{align}
where $c_l=\left|c_l\right|{\rm e}^{{\rm j}\theta_l}$, ${\tilde
x}(t)=\left|{\tilde x}(t)\right|{\rm e}^{{\rm j}\theta_{\tilde x}(t)}$ and
$\Re(\cdot)$ denotes the real part of its argument. The reference
voltage is determined by the maximum value of $\left|y(t)\right|$.

Define $\left|y\right|_{\max}$ as the maximum value of
$\left|y(t)\right|$. $|y|_{\rm max}$ is a random variable, due to the fact that the mixing function $p(t)$ randomly
alternates between $\pm 1$. Mathematically, (\ref{eq6}) indicates
that the randomness of $|y|_{\rm max}$ results from the
randomness of $\{c_l\}$, which are Fourier coefficients of $p(t)$. The threshold of $\left|y\right|_{\max}$ is
chosen such that the probability
\begin{equation}\label{eq7}
{\rm P}\left\{\left|y\right|_{\rm max}\le \left|y\right|_{\rm
th}\right\}= P_0,
\end{equation}
where $P_0$ is a constant. It is reasonable to use
$\left|y\right|_{\rm th}$ as reference voltage with a proper $P_0$,
say $99\%$.

\subsection{Main Contribution}

The following proposition reveals the property of the reference
voltage defined as (\ref{eq7}).

\begin{Proposition}\label{Pro1}
In the MWC mentioned above, consider fixed input signal (\ref{eq2}).
Further assume\\
\noindent(a) $\left|y(t)\right|$ and $\left|{\tilde
x}(t)\right|$ reach their maximum at the same time;\\
\noindent(b) $\{\theta_l\}$ are i.i.d. random variables with uniform
distribution on $[-\pi,\pi]$, and $\{\left|c_l\right|\}$ are i.i.d.
random variables independent from $\{\theta_l\}$.

Then the reference voltage $\left|y\right|_{\rm th}$ defined by
(\ref{eq7}) is proportional to $\sqrt q=\sqrt{f_s/f_p}$.
Specifically,
\begin{equation}\label{eq8}
\left|y\right|_{\rm th}=\sqrt{q}\left|{\tilde x}(t_0)\right|\sigma Y_{\rm th},
\end{equation}
where $t_0={\rm argmax}_t\left\{{\tilde x}(t)\right\}$,
$\sigma^2={\rm var}\left\{ |c_l|\cos \theta_l\right\}={\rm
var}\left\{ \Re \{c_l\}\right\}$, and $Y_{\rm th}$ satisfies
$$
2\Phi(Y_{\rm th})-1=\frac 1 {\sqrt{2\pi}} \int_{-Y_{\rm th}}^{Y_{\rm th}} {\rm e}^{-t^2/2}\,{\rm d}t=P_0.
$$
where $\Phi(x)$ is the cumulative distribution function of the standard normal distribution.
\end{Proposition}

\begin{proof}
According to assumption (a), denoting $\theta_l+2\pi lf_pt_0-\theta_{{\tilde
x}}(t_0)$ by $\hat {\theta}_l$, one has
$$
|y|_{\rm max}=\left|{\tilde x}(t_0)\right|
\left|\sum_{l=l_1}^{l_1+q-1} \left|c_l\right|\cos\left(\hat {\theta}_l\right)\right|.
$$
Notice that $\hat {\theta}_l$ is also uniform distributed on $[-\pi,\pi]$ as $\theta_l$.
Setting $z_l=|c_l|\cos \hat{\theta}_l$, the equivalent problem is to
find $Z_{\rm th}=|y|_{\rm th}/\left|{\tilde x}(t_0)\right|$ such
that
\begin{equation}\label{eq9}
{\rm P}\left\{ \left| \sum_{l=l_1}^{l_1+q-1} z_l \right|\le Z_{\rm th} \right\}= P_0.
\end{equation}

Since ${\rm E}\{\cos \hat{\theta}_l\}$ equals $0$ when $\hat{\theta_l}$ is uniform distributed on $[-\pi,\pi]$, $\{z_l\}$ are zero-mean i.i.d. random variables.
According to Central Limit Theorem \cite{ref14}, the variable
$$Z=\sum_{l=l_1}^{l_1+q-1} \frac{z_l}{\sigma \sqrt q}$$ satisfies
standard normal distribution when $q$ approaches infinity. Hence
\begin{align}\label{eq10}
{\rm P}\left\{ \left| \sum_{l=l_1}^{l_1+q-1} z_l \right|\le Z_{\rm th} \right\}
&={\rm P}\left\{ -\frac {Z_{\rm th}} {\sigma \sqrt{q}} \le Z \le \frac {Z_{\rm th}} {\sigma \sqrt{q}} \right\} \nonumber \\
&=\frac 1 {\sqrt{2\pi}} \int_{-Y_{\rm th}}^{Y_{\rm th}} {\rm e}^{-t^2/2}\,{\rm d}t,
\end{align}
where $Y_{\rm th}=Z_{\rm th}/\left(\sigma\sqrt{q}\right)$. Combining
(\ref{eq9}) and (\ref{eq10}), it is obvious that $Y_{\rm th}$ is a
constant satisfying
$$
\frac 1 {\sqrt{2\pi}} \int_{-Y_{\rm th}}^{Y_{\rm th}} {\rm e}^{-t^2/2}\,{\rm d}t=P_0.
$$
Thus
$$
|y|_{\rm th}=\left|{\tilde x}(t_0)\right|Z_{\rm th}=\sqrt{q}\left|{\tilde x}(t_0)\right|\sigma Y_{\rm th},
$$
and this completes the proof.
\end{proof}

\subsection{Discussion}

Assumption (a) is drawn from the mechanism of the system. Mapping (\ref{eq2}) to time domain, one has
$x(t)=x_0(t)\cos \left(2\pi l_0f_p t\right)$.
Thus
\begin{align}\label{eq11}
\left|{\tilde x}(t)\right|&=\left|x(t)+{\rm j}{\hat x}(t)\right| \nonumber \\
&=\left|x_0(t)\cos \left(2\pi l_0f_p t\right) + {\rm j}x_0(t)\sin \left( 2\pi l_0f_p t \right)\right| \nonumber \\
&=\left|x_0(t)\right|.
\end{align}
Equation (\ref{eq11}) shows that $\left|{\tilde x}(t)\right|$ and $\left|x_0(t)\right|$ reach the maximum at the same time. On the other hand, $y(t)$ is obtained through three steps from
$x_0(t)$: modulated by $\cos(2 \pi l_0f_p t)$ with large $l_0$
typically, multiplied by $p(t)$ and lowpass filtered. In the first step,
modulation with the carrier at high frequency can hardly change the instant when the signal reaches the maximum. The second step does not make any change because $|p(t)|=1$. For the last step, since $f_s\ge f_p$, the lowpass
filter does not affect the ``envelope'' of the signal. Based on the
analysis above, it is convincing that $|y(t)|$ and $|x_0(t)|$ also
reach the maximum at the same time, which leads to assumption (a).

Assumption (b) is based on the observation of $\{c_l,l_l\le l\le
l_1+q-1\}$. Considering (\ref{eq12}), since $l_1$ is far larger than $q$, one has $l_1+q-1\approx l_1$, which results in the approximation that
$\{c_l,l_1\le l\le l_1+q-1\}$ are i.i.d. Furthermore, the
distribution of $\theta_l$ is determined mainly by the sum of $\alpha_k\phi^{lk}, 0\le k\le M-1$. In the complex plane,
$\{\phi^{lk}, 0\le k\le M-1\}$ represents $M$ points located
uniformly on the unit circle and $M$ is typically very large. Therefore
assume that the argument of their linear combination satisfies
uniform distribution is an acceptable way to simplify the analysis.

Now we discuss about arbitrary multiband signals. Consider a multiband signal $y(t)$ which consists of $N$ basic signals denoted by (\ref{eq2}) with the same maximum value. If there is no overlap in time domain among these $N$ signals, the problem can be reduced to study one input signal. When $N$ signals reach their maximum at the same time, say $t_0$, the following analysis intuitively shows that the reference voltage is proportional to $\sqrt{Nq}$. In frequency domain, according to (\ref{eq3}), $[-f_s/2,f_s/2]$ has $q$ shift copy of each band. Thus, observing the multiband signal within a short time interval around $t_0$, the energy of $Y(f)$ is $Nq$ times the energy of each basic signal. Hence the amplitude $|y(t_0)|$ is $\sqrt{Nq}$ times the maximum of each basic signal.

\section{Simulation}
\label{sec:simu}
The following simulation calculates the reference voltage $\left|y\right|_{\rm th}$ numerically with $P_0=99\%$. For specific parameter settings, the experiment is conducted 5000 trials and the value which is larger than $99\%$ of the results is chosen as $\left|y\right|_{\rm th}$.
The number of the channels is $m=30$ and the frequency of the mixing function is $f_p\approx 51.3{\rm MHz}$. The input $x(t)$ is a multiband signal consisting of $N$ pairs of bands, defined as
\begin{equation}\label{eq13}
x(t) = \sum_{i=1}^{N}\sqrt{EB}{\rm sinc}\left(B\left(t-\tau\right)\right)\cos\left(2\pi f_i \left(t-\tau\right)\right)
\end{equation}
where $E=10$, $B=50 {\rm MHz}$, $\tau = 0.4 \rm{\mu s}$, $f_i=\left\{25if_p\right\}$. The reference voltage is calculated under the parameter choice $N=\{1,2,3\}$ and $\{q=2q'+1,\quad q'=0,1,...,9\}$.

The results are plotted in Fig.~\ref{fig4}. It is shown that $\left|y\right|_{\rm th}/ \sqrt{q}$ is almost constant when $q \le 3$. Though (\ref{eq10}) is obtained under the condition that $q$ approaches infinity, the results indicate that $q$ need not be very large. Furthermore, for each fixed $N$ respectively, calculate arithmetic mean of $\left|y\right|_{\rm th}/ \sqrt{q}$. The ratio of the results is $1:1.37:1.63$, which is close to $1:\sqrt{2}:\sqrt{3}$. The results above correspond to the theoretical analysis and the discussions.

\begin{figure}[t]
\centering
\includegraphics[width=0.35\textwidth]{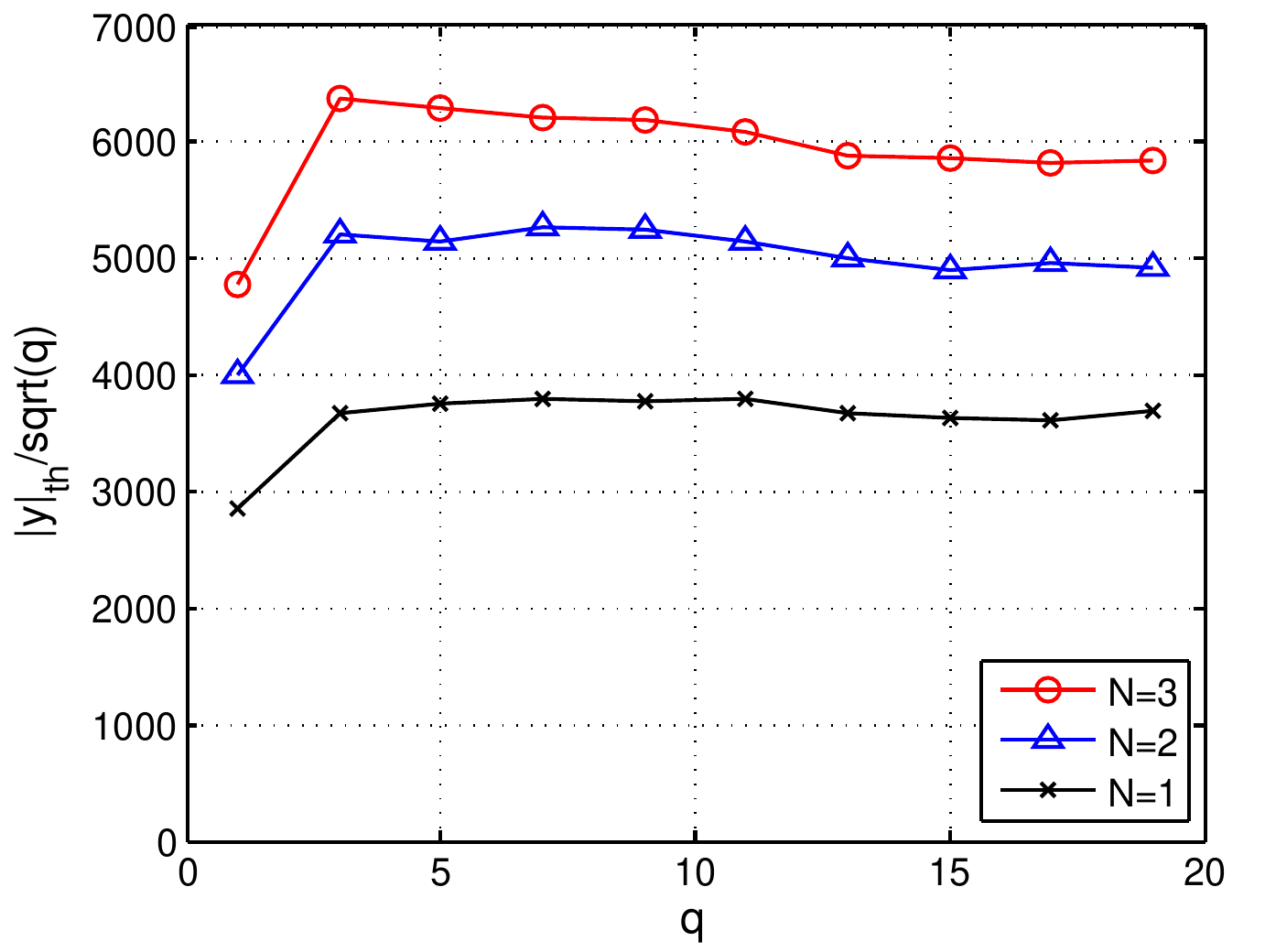}
\caption{Reference voltage with respect to collapse parameter $q$ and the number of signals $N$}\label{fig4}
\end{figure}

\section{Conclusion}
\label{sec:conclu}

In this paper, we propose a method to acquire the reference voltage of the MWC. Then we
theoretically analyze the reference voltage under rational assumptions of the input signal and the mixing functions. The conclusion is drawn that the quantization reference voltage is proportional to the square root of collapse parameter $q=f_s/f_p$. Furthermore, discussions and simulation results show that the conclusion is valid for arbitrary multiband signals. Specifically, when the multiband signal consists of $N$ pairs of bands, the reference voltage is proportional to $\sqrt{Nq}$.

\end{document}